\documentclass[a4paper,english]{article}

\usepackage{fullpage}

\usepackage[utf8]{inputenc}
\usepackage{todonotes}
\usepackage{amsmath, amssymb, amsthm}
\usepackage{hyperref}
\usepackage{xspace}
\usepackage{authblk}

\hypersetup{
    colorlinks=true,
    linkcolor=blue,
    filecolor=violet,  
    citecolor=violet,
    urlcolor=cyan,
    pdftitle={Overleaf Example},
    pdfpagemode=FullScreen,
}

\usepackage{tikz,xcolor,hyperref}

\definecolor{lime}{HTML}{A6CE39}
\DeclareRobustCommand{\orcidicon}{%
	\begin{tikzpicture}
		\draw[lime, fill=lime] (0,0) 
		circle [radius=0.16] 
		node[white] {{\fontfamily{qag}\selectfont \tiny ID}};
		\draw[white, fill=white] (-0.0625,0.095) 
		circle [radius=0.007];
	\end{tikzpicture}
	\hspace{-2mm}
}

\foreach \x in {A, ..., Z}{%
	\expandafter\xdef\csname orcid\x\endcsname{\noexpand\href{https://orcid.org/\csname orcidauthor\x\endcsname}{\noexpand\orcidicon}}
}

\newtheorem{theorem}{Theorem}
\newtheorem{lemma}{Lemma}
\newtheorem{corollary}{Corollary}
\newtheorem{definition}{Definition}

\newtheorem{remark}{Remark}

\title{On Streaming Algorithms for Geometric Independent Set and Clique}
\author[1]{Sujoy Bhore\orcidA{}}
\author[2]{Fabian Klute\orcidB{}\footnote{Supported by the FWF Austrian Science Foundation grant J-4510.}}
\author[3]{Jelle J. Oostveen\footnote{Partially supported by the NWO grant OCENW.KLEIN.114 (PACAN).}}
\affil[1]{Indian Institute of Science Education and Research, Bhopal, India\\ 
\texttt{sujoy.bhore@gmail.com}}
\affil[2]{Utrecht University, Utrecht, the Netherlands\\\authorcr
\mbox{\texttt{\{f.m.klute,j.j.oostveen\}@uu.nl}}}

\newcommand{\mc}[1]{\ensuremath{\mathcal #1}\xspace}
\newcommand{\inter}[1]{\ensuremath{\mathcal G(#1)}\xspace}

\newcommand{\indset}{maximum independent set\xspace}
\newcommand{\clique}{maximum clique\xspace}
\newcommand{\geomindset}{geometric \indset}
\newcommand{\geomclique}{geometric \clique}

\newcommand{\Indset}{Maximum independent set\xspace}

\begin{document}

\maketitle

\begin{abstract}
    % Given a large set of data it can be infeasible to even repeatedly read the input.
    % In the streaming model we are allowed only a number of constant passes over the data.
    We study the maximum geometric independent set and clique problems in the streaming model.
    Given a collection of geometric objects arriving in an insertion only stream, the aim is to find a subset such that
    all objects in the subset are pairwise disjoint or intersect respectively.
    % \sujoy{maybe we should mention how it is given}
    
    We show that no constant factor approximation algorithm exists
    to find a maximum set of independent segments or $2$-intervals without
    using a linear number of bits.
    Interestingly, our proof only requires a set of segments whose intersection 
    graph is also an interval graph.
    This reveals an interesting discrepancy between segments and intervals
    as there does exist a $2$-approximation for finding an independent set of intervals
    that uses only $O(\alpha(\mc I)\log |\mc I|)$ bits of memory for
    a set of intervals $\mc I$ with $\alpha(\mathcal I)$ being the size of the largest independent set of $\mc I$.
    % \sujoy{Should we add the citation?}.
    % there exists
    % constant-factor approximation algorithms using less than a linear number of bits of memory.
    % Our reduction for segments holds even if the intersection graph of the segments is an interval graph.
    % In contrast, it is known that the independent set can be $2$-approximated.
    % Consequently, it is not possible to convert a stream of segments into a stream of intervals
    % without essentially storing all segments.
    % This is the first such difference for streams of goemetric objects.
    On the flipside we show that for the geometric clique problem there is no constant-factor approximation
    algorithm using less than a linear number of bits even for unit intervals.
    On the positive side we show that the maximum geometric independent set in a set of 
    axis-aligned unit-height rectangles can be $4$-approximated using only $O(\alpha(\mc R)\log |\mc R|)$ bits.
    
    % \keywords{%
    %     Geometric Independent Set \and 
    %     Streaming Algorithms \and 
    %     Geometric Intersection Graphs \and 
    %     Communication Lower Bounds
    % }
\end{abstract}

\section{Introduction}
The independent set problem is one of the fundamental combinatorial optimization problems in theoretical computer science, with a wide range of applications. Given a graph $G=(V,E)$, a set of vertices $M \subset V$ is \emph{independent} if no two vertices in $M$ are adjacent in $G$. 
% A \emph{maximal independent set} (\textsc{MIS}) is an independent set that is not a proper subset of any other independent set.\fabian{Do we need the maximal ind-set, if not we could use MIS for maximum ind-set.}
A \emph{maximum independent set} is a maximum cardinality independent set. \Indset is one of the most well-studied algorithmic problems and is one of Karp’s 21 classic \textsf{NP}-complete problems~\cite{Karp72}. Moreover, it is well-known to be hard to approximate: no polynomial time
algorithm can achieve an approximation factor $n^{1-\epsilon}$, 
for any constant $\epsilon>0$, unless \textsf{P} = \textsf{NP}~\cite{Hastad96,Zuckerman07}. 
% However, computing a  can easily be done by a simple greedy algorithm in
% $O(|E|)$ time. 
\Indset serves as a natural model for many real-life optimization problems, including map labeling, computer vision, information retrieval, and scheduling; see~\cite{AgarwalKS98,BalasY86,BevernMNW15}.

\paragraph{Geometric Independent Set.}
% Independent sets on various geometric objects have been extensively studied over the years due to its vast body of applications in VLSI design, map-labeling, data mining, computer graphics, etc. 
% \fabian{Some of the applications are repeated from the previous sentence, perhaps we can merge them?}
% \fabian[inline]{Introduce geometric independent set \geomindset}
In the geometric setting we are given a set of geometric objects $\mc S$, and we say a subset $\mc S' \subseteq \mc S$ is \emph{independent} if no two objects in $\mc S'$ intersect and
we say $\mc S'$ is a \emph{clique} if every two objects pairwise intersect.
Let $\alpha(\mc S)$ be the cardinality of the largest subset $\mc S' \subseteq \mc S$ such that 
$\mc S'$ is an independent set and
$\omega(\mc S)$ the cardinality of the largest subset of $\mc S' \subseteq \mc S$ such that $\mc S'$ is a clique.
The \emph{geometric maximum independent set} and \emph{geometric maximum clique} problem ask for a set $S\subseteq \mc S$ of independent objects (that induce a clique) such that $S = \alpha(\mc S)$ ($S=\omega(\mc S)$).

% Our goal is it to compute the largest independent set or the largest clique in a stream of geometric objects.
Given a set $\mathcal{S}$ of geometric objects, we define the \emph{geometric intersection graph}
$\mc G(S)$ as the simple graph in which each vertex corresponds to an object in $\mc S$
and two vertices are connected by an edge if their corresponding objects intersect.

% \fabianchange{
Stronger results are known for the \geomindset problem is known in comparison to general graphs. A fundamental problem is the 1-dimensional case, where all objects are intervals. This problem is also known as \emph{interval selection}
problem which has applications to scheduling and resource allocation~\cite{Bar-NoyBFNS01}.
For intervals \geomindset can be solved in $O(n\log n)$ time, by a simple greedy algorithm that sweeps the line from left to right and at each step picks the interval with the leftmost right endpoint, see e.g.~\cite{KT06}. 
In contrast, the \geomindset problem is already \textsf{NP}-hard for 
% unit square intersection graphs\fabian{Add citation},
sets of segments in the plane using only two directions~\cite{kratochvil1990independent},
or $2$-intervals~\cite{DBLP:journals/siamcomp/Bar-YehudaHNSS06}.
For some restricted classes of segment intersection graphs,
such as permutation~\cite{DBLP:journals/networks/Trotter83} and circle graphs~\cite{DBLP:journals/networks/Gavril73},
the \geomindset problem can be solved in polynomial time.
Efficient approximation algorithms exist for example for unit square intersection graphs~\cite{ErlebachJS05} and more generally for pseudo disks~\cite{ChanH12}, 
as well as segments~\cite{DBLP:journals/comgeo/AgarwalM06,fox2011computing}.
In a recent breakthrough work~\cite{Mitchell21}, it was shown that there exists a constant-factor approximation scheme for \indset for a set of axis-aligned rectangles.
Very recently, this factor was improved to $3$~\cite{GalvezKMMPW22}.
% , however, whether the problem admits a \textsf{PTAS} still remains as an open problem. 
Also, the \geomindset problem has been extensively studied for dynamic geometric objects, i.e., objects can be inserted and deleted~\cite{BCIK20,CardinalIK21,CMR20,GavruskinKKL15,Henzinger0W20}.
% }
% Henzinger et al.~\cite{Henzinger0W20} studied \indset for intervals, hypercubes and hyperrectangles. Gavruskin et al.~\cite{GavruskinKKL15} considered the interval case under the assumption that no interval is fully contained in another interval and obtained an optimal solution
% with $O(\log n)$ amortized update time.
% Bhore et al.~\cite{BCIK20} showed that for intervals a $(1+\varepsilon)$-approximate \indset can be maintained with logarithmic worst-case update time. 
% Subsequently, Compton et al.~\cite{CMR20} simplified and improved this result. 
% Moreover, Bhore et al.~\cite{BCIK20} and Cardinal et al.~\cite{CardinalIK21}
% explored data structures for maintaining (expected) constant factor approximations
% for fat geometric objects such as unit squares.
% showed how the interval structure can be used to design a data structure for maintaining an expected constant factor approximate \indset of axis-aligned squares
% in the plane, with polylogarithmic amortized update time, which generalizes to $d$-dimensional hypercubes. Recently, Cardinal et al.~\cite{CardinalIK21} 
% presented a data structure that maintains a constant-factor
% approximate maximum independent set for broad classes of fat objects in $d$ dimensions, where $d$
% is assumed to be a constant.

\paragraph{Streaming algorithms.}
In this paper, we study the \geomindset and \geomclique problems for \emph{insertion only streams} of geometric objects. 
In the streaming model we consider data that is too large to fit at once into the working memory.
Instead the data is dealt with in a data stream in which we receive the elements of the input one after another in no specific order and
have access to only a limited amount of memory.
More specifically, in this model, we have bounds on the amount of available memory. As the data arrives sequentially, and we are not allowed to look at input
data of the past, unless the data was stored in our limited memory. This is effectively equivalent to assuming that we can only make one or a few passes over the input data.
We refer to~\cite{DBLP:journals/sigmod/McGregor14,muthukrishnan2005data} and the lecture notes of Chakrabati~\cite{chakrabartics49} for an overview on the general topic of streaming algorithms. 
% \todo{This is confusing in combination with `passes' over the stream (multi-pass algorithms)} 

For maximum independent set Halld\'{o}rsson et al.~\cite{HalldorssonHLS10} studied the problem for graphs and hypergraps in linear space in the semi-streaming model: Their model work in poly-logarithmic space, like in the case of the classical streaming model, but they can access and update the output buffer, treating it as an extra piece of memory. Kane et al.~\cite{KaneNW10} gave the first optimal algorithm for estimating the number of distinct elements in a data stream.

Streaming algorithms for geometric data have seen a flurry of results in recent years; see~\cite{AgarwalKMV03,ChenJLW22,CZKV20,FrahlingS05,Indyk04}. 
Note that one can also view a stream of geometric objects $\mc S$ as a vertex stream,
also called \emph{implicit vertex stream},
of its associated geometric intersection graph $\inter{\mc S}$~\cite{cormodeIndependentSets2019}.
Finding an independent, i.e. disjoint, set of geometric objects in a data stream has  been among the most studied problems in this geometric direction. 
Emek et al.~\cite{emekSpaceConstrainedInterval2016} studied the interval selection problem, where the input is a set of intervals $\mathcal{I}$ with real endpoints, and the objective is to find an independent subset of largest cardinality. They studied the interval selection problem using $O(\alpha(\mathcal{I}))$ space. They presented a $2$-approximation algorithm for the case of arbitrary intervals and a $(3/2)$-approximation for the case of unit intervals, i.e., when all intervals have the same length.
These bounds are also known to be the best possible~\cite{emekSpaceConstrainedInterval2016}.
Cabello et al.~\cite{cabelloIntervalSelection2017} studied the question of estimating $\alpha(\mc I)$ for a set $\mc I$ of intervals and gave simpler proofs of the algorithms presented by Emek et al.~\cite{emekSpaceConstrainedInterval2016}.

%Up to know mostly intervals~\cite{cabelloIntervalSelection2017} and 
%unit balls in different norms have been studied~\cite{cormodeIndependentSets2019}

% Moreover, it has been established for interval graphs that there are streaming algorithms that in one pass find a $2$-approximation of the independent set 
% using $O(\alpha(\mathcal I))$ space~\cite{cabelloIntervalSelection2017} and
% that this is also the best possible~\cite{cabelloIntervalSelection2017,HalldorssonHLS10}.
% When considering unit intervals this factor can be improved to $\frac 32$~\cite{cabelloIntervalSelection2017}. 
Cormode et al.~\cite{cormodeIndependentSets2019} considered unit balls in the $L^1$ and $L^\infty$ norms, i.e. squares in $\mathbb R^2$. 
For a set of such unit balls $\mc B$ they obtained a $3$-approximation using $O(\alpha(\mc B))$ space and 
show that there is no $\frac 52 - \varepsilon$ approximation using $o(|\mc B|)$ space.
Finally, Bakshi et al.~\cite{DBLP:conf/approx/BakshiCW20} considered Turnstile streams, i.e, deletion is also allowed, of (weighted) unit intervals and disks.

%Computing independent sets in the streaming setting has also been studied for abstract graphs.
%\fabian[inline]{Add more results.}

\subsection{Our Results}
In this paper we investigate several geometric objects that have not been 
studied in the context of streaming algorithms. 
% Specifically we look at streams of segments, $c$-intervals, and rectangles 
% for finding independent sets and sets of mutually intersecting objects.
% Our results reveal an interesting difference between segments and intervals.
% Namely, we show that given a set of $n$ segments whose intersection graph is also an interval graph,
% it is not possible to convert the segments to intervals without using $\Omega(n)$ bits of memory, i.e.,
% storing essentially all segments.
%
We show in Section~\ref{sec:segments} that there is no constant-factor approximation in the streaming model
for finding an independent set of $n$ segments using $o\left(\frac np\right)$ bits of memory
for any constant number $p$ of passes and
this bound holds even if the endpoints of the segments are on two parallel lines.
In other words, our bound holds even when the geometric intersection graph of the segments is a permutation graph.

Our construction leads to an interesting consequence.
Namely, the intersection graph created in our reduction is not only a permutationa graph, but also an interval graph and
the cardinality of its maximum independent set is not dependent on the input size.
Since there exists a $2$-approximation algorithm
for geometric independent set of a set of intervals $\mathcal I$ in the streaming model that
uses only $O(\alpha(\mathcal I))$ space
this implies that there is a difference between an interval graph being streamed as a set of intervals or as a set of segments.
We discuss this implication in Section~\ref{sec:intvssegs}.
In Section~\ref{sec:cintervals} we show that for streams of $2$-intervals 
there is no one-pass algorithm that achieves a constant-factor approximation using less than $o(n)$ bits.
On the positive side we show in \ref{sec:rectangles} that 
for $n$ axis-aligned unit height rectangles there exists a one pass streaming algorithm
achieving a $4$-approximation of the largest set of disjoint rectangles using $O(\alpha(\mathcal R)\log n)$ bits.

Finally, we show in Section~\ref{sec:clique} that the distinction between segments and intervals 
observed for the geometric independent set problem does not occur for the same objects in the geometric clique problem by
showing that there does not exist a $p$-pass algorithm using less than $o\left(\frac np\right)$ bits of memory and
achieving a constant-factor approximation of the geometric clique problem in streams of $n$ unit intervals.
We complement this hardness result by showing how to obtain an exact solution 
for the geometric clique problem in streams of $n$ intervals using only $O(n \log\omega(\mathcal I))$ bits of memory.

% \subsection{Related Work}
%\todo{S: Add dynamic results}
% \fabian[inline]{Maybe the following is too detailed for the amount of space we have in this paper?}
% \paragraph{Dynamic Independent Set.} \textsc{MIS} problem on dynamic graphs with edge updates has attracted significant attention in recent years~\cite{AssadiOSS19,BehnezhadDHSS19,ChechikZ19}.
% On the other hand, \indset problem has been extensively studied for dynamic geometric objects with vertex updates. 
% Henzinger et al.~\cite{Henzinger0W20} studied \indset for intervals, hypercubes and hyperrectangles. In another related work, Gavruskin et al.~\cite{GavruskinKKL15} considered the interval case under the assumption that no interval is fully contained in other interval and obtained an optimal solution
% with $O(\log n)$ amortized update time.
% Bhore et al.~\cite{BCIK20} showed that for intervals a $(1+\varepsilon)$-approximate \indset can be maintained with logarithmic worst-case update time. 
% Subsequently, Compton et al.~\cite{CMR20} simplified and improved this result. 
% Moreover, Bhore et al.~\cite{BCIK20} showed that how the interval structure can be used to design a data structure for maintaining an expected constant factor approximate \indset of axis-aligned squares
% in the plane, with polylogarithmic amortized update time, which generalizes to $d$-dimensional hypercubes. Lately, Cardinal et al.~\cite{CardinalIK21} 
% presented a data structure that maintains a constant-factor
% approximate maximum independent set for broad classes of fat objects in $d$ dimensions, where $d$
% is assumed to be a constant. 

\section{Independent Sets in Streams of Segments}
\label{sec:segments}
In this section we establish our lower bound for the memory necessary
to approximate the \indset problem to any constant factor on streams of segments.
We employ a lower bound reduction technique that uses multi-party set disjointness, which gives us space bounds not only for single-pass algorithms,
but also for multi-pass algorithms. 
The following problem was first studied by Alon, Matias, and Szegedy~\cite{AlonSpaceComplexity}.

\begin{definition}[\textsc{Multi-Party Set Disjointness}]
    There are $t$ players $P_1,\ldots,P_t$. Each player $P_i$ has a size $n$ bit string $x^i$. The players want to find out if there is an index $j\in [n]$ where $x^i_j = 1$ for all $i$. So, $\textsc{Disj}_{n,t}(x^1,\ldots,x^t) = \bigvee_{j=1}^n \bigwedge_{i=1}^t x^i_j$.
\end{definition}

In our proof we are going to make use of the following result on the communication complexity of \textsc{Multi-Party Set Disjointness}.
% It is known about the communication complexity 
% necessary for the players to decide the answer, even allowing randomisation.

\begin{theorem}[Chakrabarti et al.~\cite{CharkabartiSetDisjointness}]\label{thm:DisjLowerbound}
    For an error probability $0 < \delta < 1/4$, to decide $\textsc{Disj}_{n,t}$ the players need $\Omega(\frac{n}{t\log t})$ bits of communication, even for a family of instances $(x^1,\ldots,x^t)$ satisfying the following properties
    \begin{align}
        |\{ j : x^i_j = 1\}| &= n/2t& \quad \forall i\in [t] \label{property:n/2t1s}\\
        |\{ i : x^i_j = 1\}| &\in \{0,1,t\}& \quad \forall j \in [n] \label{property:1ortoverlap}\\
        |\{ j : |\{ i : x^i_j = 1\}| = t \}| &\leq 1& \label{property:leq1overlap}
    \end{align}
\end{theorem}

% Theorem~\ref{thm:DisjLowerbound} is applied by e.g. Agarwal et al.~\cite{AgarwalSpatialScan}.

We can use Theorem~\ref{thm:DisjLowerbound} by having $t$ players use a streaming algorithm to answer $\textsc{Disj}_{n,t}$. The players construct the stream by creating some part of the stream and giving it to the algorithm, and then passing the memory state of the algorithm to the next player who does the same. 
This way, the space used by the algorithm must abide to the lower bound on the communication between the players. 
We can use the $t$ players (rather than just 2) to create a bigger gap between the yes and no answer, excluding the possibility for any constant factor approximation algorithms.

\begin{figure}[t]
    \centering
    \includegraphics[page=2]{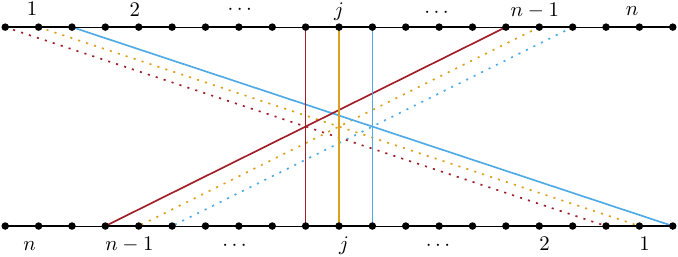}
    \caption{Lower bound for Independent Set in permutation graphs, with $t = 3$ players. Here the independent set has size $t$, by the $j$-th group, where $x^i_j=1$ for all $i\in[t]$.}
    \label{fig:IS_Perm_lowerbound}
\end{figure}

% \fabian[inline]{Maybe we should say even more explicit that $t$ depends only on the approximation ratio and not at all on $n$.}

\begin{lemma}\label{lemma:ISlowerboundPerm}
    For any $p\geq 1, t \geq 2$, any algorithm for \geomindset that can distinguish between an independent set of size $1$ and $t$ and succeeds with probability at least $3/4$ on segment streams using $p$ passes must use at least $\Omega(\frac{n}{p \cdot t\log t})$ bits of memory, even when the segment endpoints lie on two parallel lines.
\end{lemma}
\begin{proof}
    Let $(x^1,\ldots,x^t)$ be an instance of \textsc{Disjointness} with $t$ players, each with $n$ bits. We construct a permutation graph depending on the input to \textsc{Disjointness}, as illustrated in Figure~\ref{fig:IS_Perm_lowerbound} for $t=3$. Let the permutation graph have $n$ groups of $t$ points both on the top, labelled $1,\ldots, n$ from left to right. On the bottom do the same, but we label from $n$ to $1$ from left to right. For $i
    \in [t], j \in [n]$, player $i$ creates a segment from the $i$-th point in group $j$ at the top to the $i$-th point in group $j$ at the bottom when $x^i_j = 1$. This creates a permutation graph with $n' = n/2$ vertices by Property~\ref{property:n/2t1s} of Theorem~\ref{thm:DisjLowerbound}.
    
    The players construct the segment stream for some algorithm for \textsc{Max-Clique} as follows, starting from player $1$, player $i$ inputs all their $n/2t$ segments, then passes the memory state of the algorithm to player $i+1$, and this continues until all players have input their segments.
    
    We claim that the graph will contain an independent set of size $t$ if exactly $\textsc{Disj}_{n,t}(x^1,\ldots,x^t) = 1$, and otherwise the maximum independent set size is 1.
    
    First notice that any segment inserted to a group $j\in [n]$ intersects all other segments in the graph, except for segments inserted to group $j$, as these segments are parallel to it. Hence, any independent set can only contain vertices that correspond to segments inserted to the same group $j$, for some $j\in [n]$, and the size of the independent set is the number of $1$'s present over all players at index $j$. Now by Property~\ref{property:1ortoverlap} of Theorem~\ref{thm:DisjLowerbound}, any independent set can have only size $1$ or $t$ in the graph. And indeed, an independent set of size $t$ implies that all $t$ players inserted a segment for some group $j\in[n]$, and hence all have a $1$ for index $j$.
    
    Now it follows from Theorem~\ref{thm:DisjLowerbound} that any algorithm for \indset on a permutation graph with $n'$ vertices that can discern between independent set size $1$ and $t$ with probability at least $3/4$ using $p$ passes over the stream must use at least $\Omega(\frac{n}{p \cdot t\log t}) = \Omega(\frac{n'}{p \cdot t\log t})$ bits of memory.
\end{proof}

We now use Lemma~\ref{lemma:ISlowerboundPerm} to give a general hardness statement for approximation \geomindset in segment streams.

\begin{theorem}\label{thm:ISlowerbound}
    Any constant-factor approximation algorithm for \geomindset that succeeds with probability at least $3/4$ on segment streams using $p$ passes must use at least $\Omega(n/p)$ bits of memory, even when it is known that the segments correspond to a permutation graph.
\end{theorem}
\begin{proof}
    Let us be given some constant-factor approximation algorithm for \geomindset that succeeds with probability at least $3/4$ on segment streams of permutation graphs in $p$ passes. Then there exists some $c \geq 2$ such that the algorithm can distinguish between an independent set of size $1$ or $c$ in a given graph. But then we can apply Lemma~\ref{lemma:ISlowerboundPerm} to get that this algorithm must use at least $\Omega(\frac{n}{p \cdot c\log c}) = \Omega(n/p)$ bits of memory.
\end{proof}

% Notice that hardness for segment streams corresponding to permutation graphs immediately implies hardness for segment streams corresponding to circle graphs.

\section{Intervals and Segments are different}
\label{sec:intvssegs}
When we consider the intersection graph $G$ of the set of segments constructed in 
the proof of Theorem~\ref{thm:ISlowerbound}
one can see that it has a straight-forward representation as a set of intervals 
whose intersection graph is isomorphic to $G$.
Also, notice that the size of the independent set of the construction is not dependent on
the length of the bit strings, but only on the number of players.
For example, we can rule out the existence of a $2$-approximation streaming algorithm 
using any constant number $p$ of passes and $o(n/p)$ bits of memory,
already when $t/2 \geq 2 \Longleftrightarrow t \geq 4$ players are used in the construction 
presented in Lemma~\ref{lemma:ISlowerboundPerm}.

At the same time there exists a $2$-approximation one pass streaming algorithm
for independent sets of intervals that
uses only $O(\alpha(\mathcal I)\log |\mathcal I|)$ bits of memory 
where $\mathcal I$ is the set of input intervals~\cite{cabelloIntervalSelection2017,emekSpaceConstrainedInterval2016}.
Hence, these algorithms find a $2$-approximation of the independent set of the intersection graph
constructed in the proof of Lemma~\ref{lemma:ISlowerboundPerm} using only $O(\log n)$ bits of memory
if the graph was given as a set of intervals.
This leads to the following corollary.
\begin{corollary}
    Given a stream of segments $\mc S$ whose intersection graph $\inter{\mc S}$
    is in the intersection of permutation and interval graphs,
    there is no algorithm that uses $o(n/p)$ bits of memory and $p \geq 1$ passes and
    computes a stream of intervals $\mc I$ such that $\inter{\mc I}$ is isomorphic to $\inter{\mc S}$.
\end{corollary}

\section{Independent Sets in Streams of $c$-Intervals}
\label{sec:cintervals}

A \emph{$c$-interval} is a set of non-overlapping intervals $\{I_1,\ldots,I_n\}$ on the real line.
We say to $c$-intervals intersect if at least two of their intervals have one point in common.
We call a family of $c$-intervals \emph{separated} if the intervals can be split into
independent groups of intervals, each containing at most one interval from each $c$-interval,
without changing which $c$-intervals intersect.
Since we only consider $2$-intervals we can talk of left and right intervals.
Formally, let $T = \{L,R\}$ be a $2$-interval such that 
the startpoint of $L$ is left of the startpoint of $R$,
then we denote $L$ as the \emph{left} interval of $T$ and $R$ as the \emph{right} interval of $T$.

For the reduction we use the $\textsc{Chain}_t$ communication problem
which was introduced by Cormode et al.~\cite{cormodeIndependentSets2019}.

\begin{definition}[Cormode et al.~\cite{cormodeIndependentSets2019}]
    The $t$-party chained index problem \textsc{Chain$_t$} consists of 
    $t-1$ $n$-bit binary vectors $\{x^{i}\}^{t-1}_{i=1}$,
    along with corresponding indices $\{\sigma_i\}^{t-1}_{i=1}$ from the range $[n]$.
    We have the promise that the entries $\{x^{i}_{\sigma_i}\}^{t-1}_{i=1}$ 
    are all equal to the desired bit $z \in \{0,1\}$. 
    The input is initially allocated as follows:
    \begin{itemize}
        \item The first party $P_1$ knows $x^{1}$
        \item Each intermediate party $P_p$ for $1 < p < k$ knows $x^{p}$ and $\sigma_{p-1}$
        \item The final party $P_k$ knows just $\sigma_{k-1}$
    \end{itemize}
    Communication proceeds as follows: $P_1$ sends a single message to $P_2$,
    then $P_2$ communicates to $P_3$, and so on, with each party sending exactly one message 
    to its immediate successor.
    After all messages are sent, $P_k$ must correctly output $z$,
    succeeding with probability at least $2/3$.
    If the promise condition is violated, any output is considered correct.
\end{definition}

Cormode et al.~\cite{cormodeIndependentSets2019} showed the following result in the same paper.

\begin{theorem}[Cormode et al.~\cite{cormodeIndependentSets2019}]\label{thm:chain}
    Any communication scheme $\mathcal B$ which solves \textsc{Chain$_t$} 
    must communicate at least $\Omega(\frac{n}{t^2})$ bits.
\end{theorem}

In the following we use the \textsc{Chain}$_t$ to show that there is no one-pass streaming algorithm
that computes a constant-factor approximation of the maximum independent set 
for a family of separated $2$-intervals using $o(n)$ bits of memory.

\begin{remark}
For \textsc{Chain}$_t$ we may assume that party 
$i > 1$ knows all indices before $\sigma_{i-1}$. 
To realize this, just assume that every party $i$ appends 
all $i-1$ previous indices to its message.
This uses only $O(t\log n)$ bits in each such message
and hence $O(t^2\log n)$ bits in total.
As this is a lower order term with respect to the bound in Theorem~\ref{thm:chain}
we retain the linear communication bound of $\Omega(\frac{n}{t^2})$.
\end{remark}

We define an \emph{interval stack} as a set of intervals $I_1,\ldots,I_n$ on the real line where
first all startpoints appear in order of the indices and
then all endpoints, again in order of the interval indices.
We denote as \emph{left-gap} the space between 
the startpoint of $I_i$ and $I_{i+1}$ for $i = 1,\ldots,n-1$ and
the startpoint of $I_n$ and the endpoint of $I_1$.
Observe that any interval containing a point of the left gap of $I_i$ 
intersects all intervals $I_j$ with $1 \leq j \leq i$.
As for independent sets of segments we first show a technical lemma.

\begin{figure}[t]
    \centering
    \includegraphics{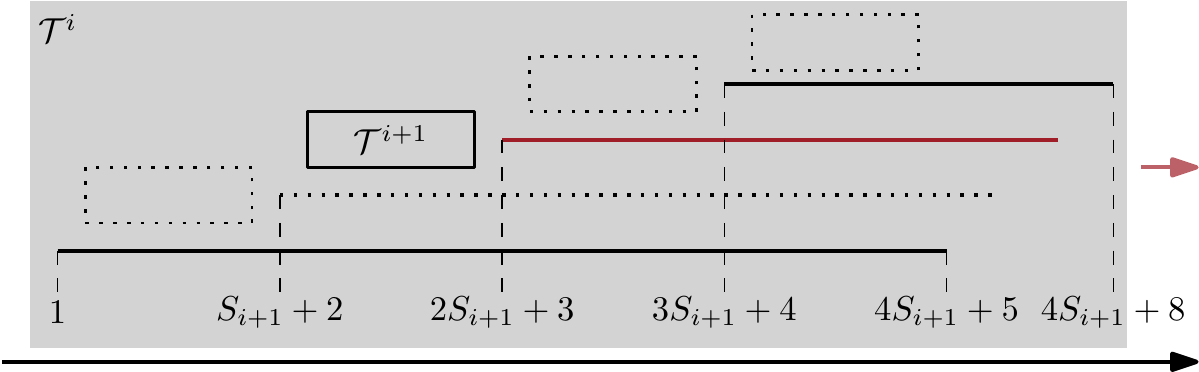}
    \caption{Left intervals created for party $i$ in the construction of Lemma~\ref{lem:twointervallower}. 
    The intervals in the grey box are for party $i$, the intervals outside are of party $i-1$,
    and the boxes mark the left gap spaces for party $i+1$.
    The red interval is the interval corresponding to a $1$ bit at $\sigma_i$.
    Dotted intervals and boxes are not actually inserted by the parties.
    The shown coordinates are the local coordinates for the interval stack inserted by party $i$.}
    \label{fig:itwointervals}
\end{figure}

\begin{lemma}\label{lem:twointervallower}
    For any $t \geq 2$, any algorithm for \geomindset that 
    can distinguish between an independent set of size $1$ and $t$ and 
    succeeds with probability at least $2/3$ on streams of $2$-intervals
    must use at least $\Omega(\frac{n}{t^3})$ bits of memory,
    even if the union of the $2$-intervals is a $2$-union interval set.
\end{lemma}
\begin{proof}
    Given an instance of \textsc{Chain}$_t$ with $\frac nt$-length bit strings $x^i$ and indices $\sigma_j$.
    Let $N = \frac nt$ and will assume for simplicity that $\frac nt$ is a whole number.
    We create one $2$-interval $T^i_j = (L^i_j,R^i_j)$ for each $1$ bit at index $j$ of bit string $i$ and
    one additional $2$-interval for party $t$.
    In the following we first describe the construction of the left intervals.
    See Figure~\ref{fig:itwointervals} for an illustration.
    
    Create an interval stack 
    $\mc L^i = \{L_j^i \mid \forall j \in [N]: x^i \text{ has a } 1 \text{ bit at index } j\}$.
    To simplify the presentation we assume that all $\frac nt$ intervals are present in $\mc T^i$.
    When actually constructing the intervals in a stream party $i$ simply does not add an interval
    when the $j$th bit is set to $0$, but still shifts the coordinates accordingly.
    We initially place the intervals of $\mc L^1$ and
    then place $\mc L^i$ for $i > 1$ in the left gap of $L^{i-1}_{\sigma_{i-1}}$.
    For player $t$ we add one interval $L_1^t$ in the left gap of $L_{\sigma_{t-1}}^{t-1}$.
    Let $\mc L$ be the union of all $\mc L^i$.
    
    To complete the construction we create the same construction using the reversed bit strings for each party.
    This creates the interval stacks $\mc R^i$, $i = 1,\ldots,n$ and with
    $R_j^i$ we denote the right interval inserted by the $i$th party for the 
    $1$ bit at index $j$ in the non-reversed bit string $x^i$.
    Let $\mc R$ be the union of all $\mc R^i$.
    Finally, we create a set of $2$-intervals as 
    $\mc T = \{(L_j^i,R_j^i) \mid L_j^i \in \mc L^i \text{ and } R_j^i \in \mc R^i \}$.
    
    Consider a $2$-interval $T_j^i = (L,R)$ inserted for bit string $x^i$ such that
    $j \neq \sigma_i$ for any $i \in \{1,\ldots,t-1\}$.
    Then, $L$ is contained in all intervals $L_b^a$ with $a < i$ and $b < j$.
    Moreover, $L$ contains every $L_b^a$ with $a > i$.
    Similarly, $R$ is contained in all intervals $R_b^a$ with $a > i$ and $b > j$ and
    contains every $R_b^a$ with $a > i$.
    Consequently if $T_j^i$ is part of an independent set $\mc I \subseteq \mc T$
    we can only add $2$-intervals $T_b^a$ to $\mc I$ with $a < i$ and $j = \sigma_a$.
    
    If the answer bit is $0$ then no $2$-interval corresponding to some $\sigma_i$ index exists and hence
    by the above argumentation the largest independent set has size one.
    If the answer bit is $1$ then there is a $2$-interval $T_{\sigma_i}^i$
    for every $\sigma_i$ with $i = 1,\ldots,t$ and all $t$ of them are independent.
    Hence, the largest independent set in this case has size $t$.
    
    It remains to describe the precise coordinates of the intervals and 
    argue that we only need $O(\log n)$ bits to represent the construction for every fixed $t$.
    For each party $i = 1,\ldots,t$ let $S_i$ be the length of the left interval stack.
    Since for party $t$ we insert only one $2$-interval we set $S_t = 2$.
    Let $\mc T$ be the set of $2$-intervals created as above. 
    In the following we consider only the left intervals of every $T_j^i \in \mc T$.
    The calculation and placement works analogously for the right intervals after reversing every $x^i$.
    Let $L_j^i$ be a left interval for party $i$ and index $j$.
    We put the startpoint of $L_j^i$ at position 
    $
        1 + (j-1) \cdot (S_{i+1} + 1)
    $
    and its endpoint at
    $
        j + N \cdot (S_{i+1} + 1).
    $
    Hence, for party $i$ its left interval stack has length at most
    \[
        S_i = N + N\cdot (S_{i+1} + 1) = N \cdot (S_{i+1} + 2).
    \]
    This can be written as a closed formula
    \[
        S_i = 4N^{t-i} \cdot 2\sum_{j=1}^{t-(i+1)}N^j = 4N^{t-i} \cdot 2\left(\frac{N^{t-i} - N}{N-1}\right).
    \]
    Now, party $i$ places its left stack at 
    \[
        P_i = 1 + \sum_{j=1}^{i-1}\left((\sigma_{j}-1) \cdot (S_{j+1} + 1) + 1\right).
    \]
    The last party places an interval of length two at position $P_{t}$.
    Since every left interval placed by a party $i > 1$ is nested by 
    the intervals inserted into the stream by the first party we 
    can conclude that $S_1 \in O\left(N^{t-1}\right)$. 
    % the maximum coordinate we ever require is if the last entry of 
    % the first party's bit string is a one, which leads to
    % \[
        % S_1 = 4N^{t-1} + 2\cdot \frac{N^{t-1}-N}{N-1} \in O\left(N^{t-1}\right).
    % \]
    This number can be represented using $O(t\log N) = O(t\log\frac nt)$ bits.
    Since $t$ can be treated as a constant we get that we only require $O(t\log\frac nt) = O(\log n)$ bits.
\end{proof}

We conclude Theorem~\ref{thm:twointervallower} from Lemma~\ref{lem:twointervallower} 
in the same way as for Theorem~\ref{thm:ISlowerbound}.
\begin{theorem}\label{thm:twointervallower}
    Any constant-factor approximation algorithm for \geomindset that succeeds with probability at least $2/3$ on streams of $2$-intervals
    requires at least $\Omega(n)$ bits of memory,
    even if the $2$-intervals are separated.
\end{theorem}

\section{Independent Sets in Streams of Unit-Height Rectangles}
\label{sec:rectangles}

In this section, we study the independent set problem for a stream of unit height arbitrary width rectangles. 
To conform with previous work we assume in this section that one cell of memory can store one rectangle, 
i.e., one cell of memory has $\Theta(\log n)$ bits where all coordinates of the rectangles are
assumed to be in $O(n)$.
Cabello and P\'erez-Lantero~\cite{cabelloIntervalSelection2017} studied the independent set problem 
for streams of intervals on the real line and achieved the following result. 

\begin{theorem}[Theorem 5~\cite{cabelloIntervalSelection2017}]
    \label{thm:2apxintvl}
    Let $\mathcal I$ be a set of intervals in the real line that arrive in a data stream.
    There is a data stream algorithm to compute a $2$-approximation to the largest independent subset of $\mathcal I$
    that uses $O(\alpha(\mathcal I))$ space and handles each interval of the stream in $O(\log \alpha(\mathcal I))$.
\end{theorem}

Using Theorem~\ref{thm:2apxintvl} we obtain a constant-factor approximation for finding the largest independent set of
rectangles in a stream of axis-aligned unit height rectangles in one pass using $O(\alpha(\mathcal R))$ space.
The below notation is similar to the one used by Cabello and P\'erez-Lantero~\cite{cabelloIntervalSelection2017}.

We divide the $y$-axis into size two intervals.
Similar to~\cite{cabelloIntervalSelection2017} we define windows $W_\ell = [\ell,\ell+2j)$ for any $j \in\mathbb Z$.
Then, we form two partitions $\mathcal W_0$ and $\mathcal W_1$ of the $y$-axis as 
$\mathcal W_z = \{W_{z + 2i} \mid i \in\mathbb Z\}$ for $z \in \{0,1\}$.
We denote with $\mathcal R_z \subseteq \mathcal R$ and $z\in\{0,1\}$ the set of rectangles that
is contained in any window of $\mathcal W_z$.
Observe, that every rectangle is fully contained in only one of the two partitions.

Computing an independent set for the rectangles $\mathcal R_z$ now amasses to computing
independent sets for each set of rectangles lying in one window $w_\ell$ of $W_z$.
By only considering windows that contain at least one interval and using Theorem~\ref{thm:2apxintvl} 
we can compute for every $W_z$ and $z \in \{0,1\}$ a 
$2$-approximation of its largest independent set using $\alpha(\mathcal R_z)$ space in one pass.
Let $\alpha'(\mathcal R_z)$ be such a $2$-approximation, 
$\alpha(\mathcal R_z)$ the size of an optimal independent set of $\mathcal R_z$, and
$\mathcal R_I \subseteq\mathcal R$ an optimal independent set of $\mathcal R$,
then it holds that
\begin{align*}
    2\max\{\alpha'(\mathcal R_0),\alpha'(\mathcal R_1)\}
    & \geq \alpha'(\mathcal R_0) + \alpha'(\mathcal R_1)
    \geq \frac12 (\alpha(\mathcal R_0) + \alpha(\mathcal R_1)) \\
    & \geq \frac12 (|\mathcal R_I \cap \mathcal R_0| + |\mathcal R_I \cap \mathcal R_1|) 
    \geq \frac12 |\mathcal R_I| \geq \frac12 \alpha(\mathcal R).
\end{align*}

From this it follows that
% \begin{align*}
    $\max\{\alpha'(\mathcal R_0),\alpha'(\mathcal R_1)\} \geq \frac14 \alpha(\mathcal R).$
% \end{align*}

\begin{theorem}
Let $\mathcal{R}$ be a set of axis-aligned unit height rectangles that arrive in a data stream,
there is an algorithm that compute a $4$-approximation to the maximum independent set of $\mathcal{R}$, 
uses $O(\alpha(\mathcal{R}))$ space, and 
handles each rectangle in polylog time.
\end{theorem}

Note, that this algorithm restricted to axis-aligned squares matches the approximation factor of three due to Cormode et al.~\cite{cormodeIndependentSets2019}
since for unit intervals we can use the $\frac{3}{2}$-approximation algorithm from Cabello et al.~\cite{cabelloIntervalSelection2017}.

\section{Clique in Streams of Intervals and Segments}
\label{sec:clique}

We can make an identical statement as Theorem~\ref{thm:ISlowerbound} for maximum clique instead of \indset by observing the complement graph of the construction in Lemma~\ref{lemma:ISlowerboundPerm}.

\begin{figure}[t]
    \centering
    \includegraphics[page=2]{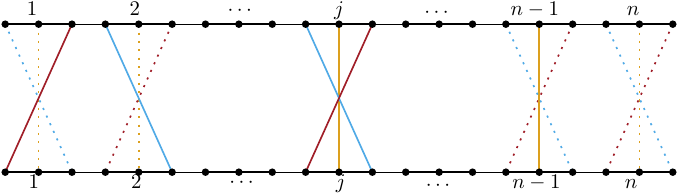}
    \caption{Lower bound for Clique in permutation graphs, with $t = 3$ players. In this example, $x^i_j = 1$ for all players $i\in[t]$.}
    \label{fig:Clique_Perm_lowerbound}
\end{figure}

\begin{theorem}\label{thm:Cliquelowerbound}
    Any constant-factor approximation algorithm for geometric maximum clique that succeeds with probability at least $3/4$ on segment streams using $p$ passes must use at least $\Omega(n/p)$ bits of memory, even when the endpoints of the segments lie on two lines.
\end{theorem}
\begin{proof}
    The complement graph of the construction of Lemma~\ref{lemma:ISlowerboundPerm} is also a permutation graph and admits the property that it contains either a clique of size $1$ or $t$. It is given by reversing the permutation (exactly mirroring the bottom) of the construction in Lemma~\ref{lemma:ISlowerboundPerm}. This construction is illustrated in Figure~\ref{fig:Clique_Perm_lowerbound}. It follows that Lemma~\ref{lemma:ISlowerboundPerm} also holds for geometric maximum clique instead of \geomindset. The theorem now follows from the proof of Theorem~\ref{thm:ISlowerbound}, using maximum clique instead of \indset.
\end{proof}

%%%% Fabian: I don't think we need this sentence, its obvious and it helps us use just 11 pages.
%
% Notice that our proofs generalizes immediately to
% hardness for segment streams corresponding to permutation graphs immediately implies hardness for segment streams corresponding to circle graphs.
%
For streams of intervals, we show a simple upper bound, using that there are at most $2n$ different endpoints of intervals.

\begin{theorem}\label{thm:CliqueUpperBoundInterval}
    Let $\mathcal{I}$ be a set of intervals in the real line that arrive in a data stream. There is an algorithm to compute the largest clique size, $\omega(\mathcal{I})$, in 1 pass using $O(n \log(\omega(\mathcal{I}))$ bits of memory, using time $O(n^2)$ total. In a second pass, the intervals that make up the clique can be recovered, which can be streamed without extra memory use, or stored using $O(\omega(\mathcal{I}) \log n)$ bits of memory.
\end{theorem}
\begin{proof}
    We keep a counter for every possible endpoint of an interval, which are $2n$ counters total. We keep the order of counters fixed, but need no labels for a counter, because of the assumption that the range of endpoints is $1,\ldots,2n$. When an interval appears in the stream, we increment all counters that are contained in the interval, including its endpoints. At the end of the stream, $\omega(\mathcal{I})$ is given by the largest counter, as this coordinate is a witness to $\omega(\mathcal{I})$ intervals co-intersecting. This is the correct maximum, as the number of intersecting intervals can only change at an endpoint of an interval.
    
    In the second pass, we can recover the intervals that make up the clique can be recovered by pushing every interval that overlaps the coordinate of the maximum counter found in the first pass to the output.
\end{proof}

The result of Theorem~\ref{thm:CliqueUpperBoundInterval} is nearly tight, as the construction of Theorem~\ref{thm:Cliquelowerbound} can also be constructed as a stream of unit intervals.
% and thus holds for interval streams as well.

% \section{Overlap and containments}
% \fabian[inline]{More a reminder to maybe integrate this somewhere}
% In the previous sections we considered two objects as independent if they have no point in common.
% Instead of considering any point we here consider two objects as independent if
% they contain each other which means only overlapping objects intersect or
% they overlap each other which means only two objects where one contains the other are considered intersecting.

% Considering the constructions used in Lemma~\ref{lemma:ISlowerboundPerm} and
% Theorem~\ref{thm:Cliquelowerbound} we can see that they can be adapted to 
% work with sets of overlapping or containing intervals.

\section{Conclusion}
\label{sec:conclusion}
We studied the geometric independent set and clique problems for a variety of geometric objects.
Interestingly, we showed that the type of geometric object used for the implicit stream of a 
geometric intersection graph can make a substantial difference even for simple objects like segments and intervals.
This raises the question if such a difference also exists for other types of objects.
Moreover, the complexity of finding an independent set in a stream of arbitrary rectangles remains open.
Finally, studying streams of geometric objects in other streaming models, such as turnstile streams, provides an interesting direction for future research.

\bibliographystyle{plainurl}
\bibliography{ref}

\begin{thebibliography}{10}

\bibitem{AgarwalKMV03}
Pankaj~K. Agarwal, Shankar Krishnan, Nabil~H. Mustafa, and Suresh
  Venkatasubramanian.
\newblock Streaming geometric optimization using graphics hardware.
\newblock In {\em Proceedings of the 11th Annual European Symposium on
  Algorithms ({ESA}'03)}, volume 2832 of {\em LNCS}, pages 544--555. Springer,
  2003.
\newblock \href {https://doi.org/10.1007/978-3-540-39658-1\_50}
  {\path{doi:10.1007/978-3-540-39658-1\_50}}.

\bibitem{DBLP:journals/comgeo/AgarwalM06}
Pankaj~K. Agarwal and Nabil~H. Mustafa.
\newblock Independent set of intersection graphs of convex objects in 2d.
\newblock {\em Computational Geometry: Theory and Applications}, 34(2):83--95,
  2006.
\newblock \href {https://doi.org/10.1016/j.comgeo.2005.12.001}
  {\path{doi:10.1016/j.comgeo.2005.12.001}}.

\bibitem{AgarwalKS98}
Pankaj~K. Agarwal, Marc~J. van Kreveld, and Subhash Suri.
\newblock Label placement by maximum independent set in rectangles.
\newblock {\em Computational Geometry: Theory and Application},
  11(3-4):209--218, 1998.
\newblock \href {https://doi.org/10.1016/S0925-7721(98)00028-5}
  {\path{doi:10.1016/S0925-7721(98)00028-5}}.

\bibitem{AlonSpaceComplexity}
Noga Alon, Yossi Matias, and Mario Szegedy.
\newblock The space complexity of approximating the frequency moments.
\newblock {\em Journal of Computer and System Sciences}, 58(1):137--147, 1999.
\newblock \href {https://doi.org/10.1006/jcss.1997.1545}
  {\path{doi:10.1006/jcss.1997.1545}}.

\bibitem{DBLP:conf/approx/BakshiCW20}
Ainesh Bakshi, Nadiia Chepurko, and David~P. Woodruff.
\newblock Weighted maximum independent set of geometric objects in turnstile
  streams.
\newblock In {\em Proceedings of the Annual International Conference on
  Approximation, Randomization, and Combinatorial Optimization. Algorithms and
  Techniques ({APPROX/RANDOM}'20)}, volume 176 of {\em LIPIcs}, pages
  64:1--64:22. Schloss Dagstuhl - Leibniz-Zentrum f{\"{u}}r Informatik, 2020.
\newblock \href {https://doi.org/10.4230/LIPIcs.APPROX/RANDOM.2020.64}
  {\path{doi:10.4230/LIPIcs.APPROX/RANDOM.2020.64}}.

\bibitem{BalasY86}
Egon Balas and Chang~Sung Yu.
\newblock Finding a maximum clique in an arbitrary graph.
\newblock {\em {SIAM} Journal on Computing}, 15(4):1054--1068, 1986.
\newblock \href {https://doi.org/10.1137/0215075} {\path{doi:10.1137/0215075}}.

\bibitem{Bar-NoyBFNS01}
Amotz Bar-Noy, Reuven Bar-Yehuda, Ari Freund, Joseph Naor, and Baruch Schieber.
\newblock A unified approach to approximating resource allocation and
  scheduling.
\newblock {\em Journal of the {ACM}}, 48(5):1069--1090, 2001.
\newblock \href {https://doi.org/10.1145/502102.502107}
  {\path{doi:10.1145/502102.502107}}.

\bibitem{DBLP:journals/siamcomp/Bar-YehudaHNSS06}
Reuven Bar{-}Yehuda, Magn{\'{u}}s~M. Halld{\'{o}}rsson, Joseph Naor, Hadas
  Shachnai, and Irina Shapira.
\newblock Scheduling split intervals.
\newblock {\em {SIAM} Journal on Computing}, 36(1):1--15, 2006.
\newblock \href {https://doi.org/10.1137/S0097539703437843}
  {\path{doi:10.1137/S0097539703437843}}.

\bibitem{BCIK20}
Sujoy Bhore, Jean Cardinal, John Iacono, and Grigorios Koumoutsos.
\newblock Dynamic geometric independent set.
\newblock {\em CoRR}, abs/2007.08643, 2020.
\newblock URL: \url{https://arxiv.org/abs/2007.08643}, \href
  {http://arxiv.org/abs/2007.08643} {\path{arXiv:2007.08643}}.

\bibitem{cabelloIntervalSelection2017}
Sergio Cabello and Pablo {P{\'e}rez-Lantero}.
\newblock Interval selection in the streaming model.
\newblock {\em Theoretical Computer Science}, 702:77--96, 2017.
\newblock \href {https://doi.org/10.1016/j.tcs.2017.08.015}
  {\path{doi:10.1016/j.tcs.2017.08.015}}.

\bibitem{CardinalIK21}
Jean Cardinal, John Iacono, and Grigorios Koumoutsos.
\newblock Worst-case efficient dynamic geometric independent set.
\newblock In {\em Proceedings of the 29th Annual European Symposium on
  Algorithms ({ESA}'21)}, volume 204 of {\em LIPIcs}, pages 25:1--25:15.
  Schloss Dagstuhl - Leibniz-Zentrum f{\"{u}}r Informatik, 2021.
\newblock \href {https://doi.org/10.4230/LIPIcs.ESA.2021.25}
  {\path{doi:10.4230/LIPIcs.ESA.2021.25}}.

\bibitem{chakrabartics49}
Amit Chakrabarti.
\newblock {CS}49: Data stream algorithms lecture notes, 2020.
\newblock URL:
  \url{http://www.cs.dartmouth.edu/ac/Teach/data-streams-lecnotes.pdf}.

\bibitem{CharkabartiSetDisjointness}
Amit Chakrabarti, Subhash Khot, and Xiaodong Sun.
\newblock Near-optimal lower bounds on the multi-party communication complexity
  of set disjointness.
\newblock In {\em 18th Annual {IEEE} Conference on Computational Complexity
  (Complexity 2003), 7-10 July 2003, Aarhus, Denmark}, pages 107--117. {IEEE}
  Computer Society, 2003.
\newblock \href {https://doi.org/10.1109/CCC.2003.1214414}
  {\path{doi:10.1109/CCC.2003.1214414}}.

\bibitem{ChanH12}
Timothy~M. Chan and Sariel Har{-}Peled.
\newblock Approximation algorithms for maximum independent set of pseudo-disks.
\newblock {\em Discrete \& Computational Geometry}, 48(2):373--392, 2012.
\newblock \href {https://doi.org/10.1007/s00454-012-9417-5}
  {\path{doi:10.1007/s00454-012-9417-5}}.

\bibitem{ChenJLW22}
Xi~Chen, Rajesh Jayaram, Amit Levi, and Erik Waingarten.
\newblock New streaming algorithms for high dimensional {EMD} and {MST}.
\newblock In {\em Proceedings of the 54th Annual {ACM} {SIGACT} Symposium on
  Theory of Computing ({STOC}'22)}, pages 222--233. {ACM}, 2022.
\newblock \href {https://doi.org/10.1145/3519935.3519979}
  {\path{doi:10.1145/3519935.3519979}}.

\bibitem{CMR20}
Spencer Compton, Slobodan Mitrovic, and Ronitt Rubinfeld.
\newblock New partitioning techniques and faster algorithms for approximate
  interval scheduling.
\newblock {\em CoRR}, abs/2012.15002, 2020.
\newblock URL: \url{https://arxiv.org/abs/2012.15002}, \href
  {http://arxiv.org/abs/2012.15002} {\path{arXiv:2012.15002}}.

\bibitem{cormodeIndependentSets2019}
Graham Cormode, Jacques Dark, and Christian Konrad.
\newblock Independent sets in vertex-arrival streams.
\newblock In {\em Proceedings of the 46th International Colloquium on Automata,
  Languages, and Programming, ({ICALP}'19)}, volume 132 of {\em LIPIcs}, pages
  45:1--45:14. Schloss Dagstuhl - Leibniz-Zentrum f{\"{u}}r Informatik, 2019.
\newblock \href {https://doi.org/10.4230/LIPIcs.ICALP.2019.45}
  {\path{doi:10.4230/LIPIcs.ICALP.2019.45}}.

\bibitem{CZKV20}
Artur Czumaj, Shaofeng~H.{-}C. Jiang, Robert Krauthgamer, and Pavel
  Vesel{\'{y}}.
\newblock Streaming algorithms for geometric steiner forest.
\newblock {\em CoRR}, abs/2011.04324, 2020.
\newblock URL: \url{https://arxiv.org/abs/2011.04324}, \href
  {http://arxiv.org/abs/2011.04324} {\path{arXiv:2011.04324}}.

\bibitem{emekSpaceConstrainedInterval2016}
Yuval Emek, Magn{\'{u}}s~M. Halld{\'{o}}rsson, and Adi Ros{\'{e}}n.
\newblock Space-constrained interval selection.
\newblock {\em {ACM} Transactions on Algorithms}, 12(4):51:1--51:32, 2016.
\newblock \href {https://doi.org/10.1145/2886102} {\path{doi:10.1145/2886102}}.

\bibitem{ErlebachJS05}
Thomas Erlebach, Klaus Jansen, and Eike Seidel.
\newblock Polynomial-time approximation schemes for geometric intersection
  graphs.
\newblock {\em {SIAM} Journal on Computing}, 34(6):1302--1323, 2005.
\newblock \href {https://doi.org/10.1137/S0097539702402676}
  {\path{doi:10.1137/S0097539702402676}}.

\bibitem{fox2011computing}
Jacob Fox and J{\'a}nos Pach.
\newblock Computing the independence number of intersection graphs.
\newblock In {\em Proceedings of the 22nd Annual ACM-SIAM Symposium on Discrete
  Algorithms ({SODA}'11)}, pages 1161--1165. SIAM, 2011.

\bibitem{FrahlingS05}
Gereon Frahling and Christian Sohler.
\newblock Coresets in dynamic geometric data streams.
\newblock In {\em Proceedings of the 37th Annual {ACM} Symposium on Theory of
  Computing ({STOC}'05)}, pages 209--217. {ACM}, 2005.
\newblock \href {https://doi.org/10.1145/1060590.1060622}
  {\path{doi:10.1145/1060590.1060622}}.

\bibitem{GalvezKMMPW22}
Waldo G{\'{a}}lvez, Arindam Khan, Mathieu Mari, Tobias M{\"{o}}mke,
  Madhusudhan~Reddy Pittu, and Andreas Wiese.
\newblock A 3-approximation algorithm for maximum independent set of
  rectangles.
\newblock In {\em Proceedings of the 2022 {ACM-SIAM} Symposium on Discrete
  Algorithms, ({SODA}'22)}, pages 894--905. {SIAM}, 2022.
\newblock \href {https://doi.org/10.1137/1.9781611977073.38}
  {\path{doi:10.1137/1.9781611977073.38}}.

\bibitem{DBLP:journals/networks/Gavril73}
Fanica Gavril.
\newblock Algorithms for a maximum clique and a maximum independent set of a
  circle graph.
\newblock {\em Networks}, 3(3):261--273, 1973.
\newblock \href {https://doi.org/10.1002/net.3230030305}
  {\path{doi:10.1002/net.3230030305}}.

\bibitem{GavruskinKKL15}
Alexander Gavruskin, Bakhadyr Khoussainov, Mikhail Kokho, and Jiamou Liu.
\newblock Dynamic algorithms for monotonic interval scheduling problem.
\newblock {\em Theoretical Computer Science}, 562:227--242, 2015.
\newblock \href {https://doi.org/10.1016/j.tcs.2014.09.046}
  {\path{doi:10.1016/j.tcs.2014.09.046}}.

\bibitem{HalldorssonHLS10}
Bjarni~V. Halld{\'{o}}rsson, Magn{\'{u}}s~M. Halld{\'{o}}rsson, Elena
  Losievskaja, and Mario Szegedy.
\newblock Streaming algorithms for independent sets.
\newblock In {\em Proceedings of the 37th International Colloquium on Automata,
  Languages and Programming ({ICALP}'10)}, volume 6198 of {\em LNCS}, pages
  641--652. Springer, 2010.
\newblock \href {https://doi.org/10.1007/978-3-642-14165-2\_54}
  {\path{doi:10.1007/978-3-642-14165-2\_54}}.

\bibitem{Hastad96}
Johan H{\aa}stad.
\newblock Clique is hard to approximate within $n^{1-\epsilon}$.
\newblock In {\em Proceedinsg of the 37th Annual Symposium on Foundations of
  Computer Science ({FOCS}'96)}, pages 627--636. {IEEE} Computer Society, 1996.
\newblock \href {https://doi.org/10.1109/SFCS.1996.548522}
  {\path{doi:10.1109/SFCS.1996.548522}}.

\bibitem{Henzinger0W20}
Monika Henzinger, Stefan Neumann, and Andreas Wiese.
\newblock Dynamic approximate maximum independent set of intervals, hypercubes
  and hyperrectangles.
\newblock In {\em Proceedings of the 36th International Symposium on
  Computational Geometry ({SoCG}'20)}, volume 164 of {\em LIPIcs}, pages
  51:1--51:14. Schloss Dagstuhl - Leibniz-Zentrum f{\"{u}}r Informatik, 2020.
\newblock \href {https://doi.org/10.4230/LIPIcs.SoCG.2020.51}
  {\path{doi:10.4230/LIPIcs.SoCG.2020.51}}.

\bibitem{Indyk04}
Piotr Indyk.
\newblock Streaming algorithms for geometric problems.
\newblock In {\em Proceedings of the 24th International Conference on
  Foundations of Software Technology and Theoretical Computer Science
  ({FSTTCS}'04)}, volume 3328 of {\em LNCS}, pages 32--34. Springer, 2004.
\newblock \href {https://doi.org/10.1007/978-3-540-30538-5\_3}
  {\path{doi:10.1007/978-3-540-30538-5\_3}}.

\bibitem{DBLP:journals/networks/Trotter83}
Leslie E.~Trotter Jr.
\newblock Algorithmic graph theory and perfect graphs, by martin c. golumbic,
  academic, new york, 284 pp. price: {\textdollar}34.00.
\newblock {\em Networks}, 13(2):304--305, 1983.
\newblock \href {https://doi.org/10.1002/net.3230130214}
  {\path{doi:10.1002/net.3230130214}}.

\bibitem{KaneNW10}
Daniel~M. Kane, Jelani Nelson, and David~P. Woodruff.
\newblock An optimal algorithm for the distinct elements problem.
\newblock In {\em Proceedings of the 29th {ACM} {SIGMOD-SIGACT-SIGART}
  Symposium on Principles of Database Systems ({PODS}'10)}, pages 41--52.
  {ACM}, 2010.
\newblock \href {https://doi.org/10.1145/1807085.1807094}
  {\path{doi:10.1145/1807085.1807094}}.

\bibitem{Karp72}
Richard~M. Karp.
\newblock Reducibility among combinatorial problems.
\newblock In {\em Proceedings of a Symposium on the Complexity of Computer
  Computations}, The {IBM} Research Symposia Series, pages 85--103. Plenum
  Press, New York, 1972.
\newblock \href {https://doi.org/10.1007/978-1-4684-2001-2\_9}
  {\path{doi:10.1007/978-1-4684-2001-2\_9}}.

\bibitem{KT06}
Jon~M. Kleinberg and {\'{E}}va Tardos.
\newblock {\em Algorithm design}.
\newblock Addison-Wesley, 2006.

\bibitem{kratochvil1990independent}
Jan Kratochv{\'\i}l and Jaroslav Ne{\v{s}}et{\v{r}}il.
\newblock Independent set and clique problems in intersection-defined classes
  of graphs.
\newblock {\em Commentationes Mathematicae Universitatis Carolinae},
  31(1):85--93, 1990.

\bibitem{DBLP:journals/sigmod/McGregor14}
Andrew McGregor.
\newblock Graph stream algorithms: a survey.
\newblock {\em {SIGMOD} Record}, 43(1):9--20, 2014.
\newblock \href {https://doi.org/10.1145/2627692.2627694}
  {\path{doi:10.1145/2627692.2627694}}.

\bibitem{Mitchell21}
Joseph S.~B. Mitchell.
\newblock Approximating maximum independent set for rectangles in the plane.
\newblock In {\em Proceedings of the 62nd {IEEE} Annual Symposium on
  Foundations of Computer Science, ({FOCS}'21}, pages 339--350. {IEEE}, 2021.
\newblock \href {https://doi.org/10.1109/FOCS52979.2021.00042}
  {\path{doi:10.1109/FOCS52979.2021.00042}}.

\bibitem{muthukrishnan2005data}
Shanmugavelayutham Muthukrishnan et~al.
\newblock Data streams: Algorithms and applications.
\newblock {\em Foundations and Trends{\textregistered} in Theoretical Computer
  Science}, 1(2):117--236, 2005.

\bibitem{BevernMNW15}
Ren{\'{e}} van Bevern, Matthias Mnich, Rolf Niedermeier, and Mathias Weller.
\newblock Interval scheduling and colorful independent sets.
\newblock {\em Journal of Scheduling}, 18(5):449--469, 2015.
\newblock \href {https://doi.org/10.1007/s10951-014-0398-5}
  {\path{doi:10.1007/s10951-014-0398-5}}.

\bibitem{Zuckerman07}
David Zuckerman.
\newblock Linear degree extractors and the inapproximability of max clique and
  chromatic number.
\newblock {\em Theory of computation}, 3(1):103--128, 2007.
\newblock \href {https://doi.org/10.4086/toc.2007.v003a006}
  {\path{doi:10.4086/toc.2007.v003a006}}.

\end{thebibliography}

% \clearpage

% \section{Squares and Fat Objects}
% \fabian[inline]{Integrate this into introduction without formal statement.}

% \section{Rectangles}

% Idea - Here, we aim for a two pass streaming algorithm that can provide a $\log$-factor approximate solution with polylog storage and time. 
% We use the idea of Agarwal et al. that recursively
% partitions the instance. In the first pass, we compute all the medians that will be used to compute the independent set for the stabbing lines. 
% We use the result of Cormode et al.~\cite{cormode2009count}.

% \begin{theorem}
% Let $\mathcal{R}$ be a set of rectangles that arrive in data stream.
% There is a two-pass streaming algorithm that compute a $O(\log n)$-approximation to the maximum independent set of $\mathcal{R}$ that uses $O(...)$ space and handle each rectangle of the stream in polylog time. 
% \end{theorem}

% \paragraph{Improvement.} Can we improve this result by keeping an $\epsilon$-net. But, $\epsilon$-net may not be a good approximation to the maximum independnet set. How about the $\epsilon$-net in the dual...?

\end{document}